\documentclass[12pt]{article}

\def\UseBibLatex{1}

\makeatletter
\def\input@path{{styles/}}
\makeatother

\newcommand{\UsePackage}[1]{%
  \IfFileExists{styles/#1.sty}{%
      \usepackage{styles/#1}%
   }{%
      \IfFileExists{../styles/#1.sty}{%
         \usepackage{../styles/#1}%
      }{%
         \usepackage{#1}%
      }%
   }%
}

\usepackage[T1]{fontenc}
\usepackage{lmodern}
\usepackage{textcomp}

\usepackage{amsmath}%
\usepackage{amssymb}%
\usepackage[table]{xcolor}%

\usepackage{todonotes}
\usepackage[in]{fullpage}%

\usepackage[amsmath,thmmarks]{ntheorem}%
\theoremseparator{.}%

\usepackage{titlesec}%
\titlelabel{\thetitle. }%
\usepackage{xcolor}%
\usepackage{mleftright}%
\usepackage{xspace}%
\usepackage{graphicx}
\usepackage{hyperref}%

\usepackage{hyperref}%
\hypersetup{%
      unicode,
      breaklinks,%
      colorlinks=true,%
      urlcolor=[rgb]{0.25,0.0,0.0},%
      linkcolor=[rgb]{0.5,0.0,0.0},%
      citecolor=[rgb]{0,0.2,0.445},%
      filecolor=[rgb]{0,0,0.4},
      anchorcolor=[rgb]={0.0,0.1,0.2}%
}
\usepackage[ocgcolorlinks]{ocgx2}

\theoremseparator{.}%

\theoremstyle{plain}%
\newtheorem{theorem}{Theorem}[section]

\newtheorem{lemma}[theorem]{Lemma}
\newtheorem{conjecture}[theorem]{Conjecture}
\newtheorem{corollary}[theorem]{Corollary}

\newtheorem{observation}[theorem]{Observation}

\newtheorem{proposition}[theorem]{Proposition}

\theoremstyle{plain}%
\theoremheaderfont{\sf} \theorembodyfont{\upshape}%
\newtheorem*{remark:unnumbered}[theorem]{Remark}%
\newcommand{\myqedsymbol}{\rule{2mm}{2mm}}

\theoremheaderfont{\em}%
\theorembodyfont{\upshape}%
\theoremstyle{nonumberplain}%
\theoremseparator{}%
\theoremsymbol{\myqedsymbol}%
\newtheorem{proof}{Proof:}%

\providecommand{\emphind}[1]{}%
\renewcommand{\emphind}[1]{\emph{#1}\index{#1}}

\definecolor{blue25emph}{rgb}{0, 0, 11}

\providecommand{\emphic}[2]{}
\renewcommand{\emphic}[2]{\textcolor{blue25emph}{%
      \textbf{\emph{#1}}}\index{#2}}

\providecommand{\emphi}[1]{}%
\renewcommand{\emphi}[1]{\emphic{#1}{#1}}

\definecolor{almostblack}{rgb}{0, 0, 0.3}

\providecommand{\emphw}[1]{}%
\renewcommand{\emphw}[1]{{\textcolor{almostblack}{\emph{#1}}}}%

\providecommand{\emphOnly}[1]{}%
\renewcommand{\emphOnly}[1]{\emph{\textcolor{blue25}{\textbf{#1}}}}

\newcommand{\HLink}[2]{\hyperref[#2]{#1~\ref*{#2}}}
\newcommand{\HLinkSuffix}[3]{\hyperref[#2]{#1\ref*{#2}{#3}}}

\providecommand{\deflab}[1]{}
\renewcommand{\deflab}[1]{\label{def:#1}}

\providecommand{\eqlab}[1]{}%
\renewcommand{\eqlab}[1]{\label{equation:#1}}

\newcommand{\remove}[1]{}%

\newcommand{\ceil}[1]{\mleft\lceil {#1} \mright\rceil}
\newcommand{\floor}[1]{\mleft\lfloor {#1} \mright\rfloor}

\usepackage[inline]{enumitem}

\newlist{compactenumA}{enumerate}{5}%
\setlist[compactenumA]{topsep=0pt,itemsep=-1ex,partopsep=1ex,parsep=1ex,%
   label=(\Alph*)}%

\newlist{compactenuma}{enumerate}{5}%
\setlist[compactenuma]{topsep=0pt,itemsep=-1ex,partopsep=1ex,parsep=1ex,%
   label=(\alph*)}%

\newlist{compactenumI}{enumerate}{5}%
\setlist[compactenumI]{topsep=0pt,itemsep=-1ex,partopsep=1ex,parsep=1ex,%
   label=(\Roman*)}%

\newlist{compactenumi}{enumerate}{5}%
\setlist[compactenumi]{topsep=0pt,itemsep=-1ex,partopsep=1ex,parsep=1ex,%
   label=(\roman*)}%

\newlist{compactitem}{itemize}{5}%
\setlist[compactitem]{topsep=0pt,itemsep=-1ex,partopsep=1ex,parsep=1ex,%
   label=\ensuremath{\bullet}}%

\providecommand{\BibLatexMode}[1]{}
\providecommand{\BibTexMode}[1]{}

\ifx\UseBibLatex\undefined%
  \renewcommand{\BibLatexMode}[1]{}
  \renewcommand{\BibTexMode}[1]{#1}
\else
  \renewcommand{\BibLatexMode}[1]{#1}
  \renewcommand{\BibTexMode}[1]{}
\fi

\BibLatexMode{%
   \usepackage[bibencoding=utf8,style=alphabetic,backend=biber]{biblatex}%
   \UsePackage{my_biblatex}%
}

\numberwithin{figure}{section}%
\numberwithin{table}{section}%
\numberwithin{equation}{section}%

\BibLatexMode{%
   \bibliography{template}
}

\begin{document}

\title{Improved upper bounds for the Heilbronn's  Problem for $k$-gons}

\author{Rishikesh Gajjala\thanks{Indian Institute of Science,
		Bengaluru, {\tt rishikeshg@iisc.ac.in}}
	\and
	Jayanth Ravi}

\date{ }

\maketitle

\begin{abstract}
The Heilbronn triangle problem asks for the placement of $n$ points in a unit square that maximizes the smallest area of a triangle formed by any three of those points. In $1972$, Schmidt considered a natural generalization of this problem. He asked for the placement of $n$ points in a unit square that maximizes the smallest area of the convex hull formed by any four of those points. He showed a lower bound of $\Omega(n^{-3/2})$, which was improved to $\Omega(n^{-3/2}\log{n})$ by Leffman. 

A trivial upper bound of $3/n$ could be obtained and 
Schmidt asked if this can be improved asymptotically. However, despite several efforts, no asymptotic improvement over the trivial upper bound was known for the last $50$ years, and the problem started to get the tag of being notoriously hard. Szemer{\'e}di posed the question of whether one can, at least, improve the constant in this trivial upper bound. In this work, we answer this question by proving an upper bound of $2/n+o(1/n)$. We also extend our results to any convex hulls formed by $k\geq 4$ points.
\end{abstract}

\section{Introduction}
Given a constant $k \geq 3$ and a set $\mathcal{P}=\{P_1,P_2,\ldots,P_n\}$ of $n\geq k$ points on the unit square $[0,1]^2$, let $\mathcal{A}_{k}(P)$ be the area of the smallest convex hull among all convex hulls determined by subsets of $k$ points in $\mathcal{P}$. The supremum value of $\mathcal{A}_{k}(P)$ over all choices of $\mathcal{P}$ is denoted by $\Delta_k(n)$. The Heilbronn triangle problem asks for the value of $\Delta_3(n)$. %The exact value of $\Delta_3(n)$ is only known and for $n \leq 6$ and is open for $n \geq 7$

The Heilbronn triangle problem is one of the fundamental problems in discrete geometry and discrepancy theory and has a rich history.  Paul Erdős proved that $\Delta_3(n)=\Omega\left(\frac{1}{n^2}\right)$. This was believed to be the upper bound for some time until Koml{\'o}s, Pintz and Szemer{\'e}di \cite{komlos1982lower} proved that $$\Delta_3(n)=\Omega\left(\frac{\log{n}}{n^2}\right)$$ 
Over a series of works, the upper bounds were improved by Roth \cite{roth1972problem, roth1973estimation, roth1976developments, roth1951problem} and Schmidt \cite{10.1112/jlms/s2-4.3.545}. The current best-known upper bound is due to Koml{\'o}s, Pintz and Szemer{\'e}di \cite{komlos1981heilbronn} $$\Delta_3(n)=\mathcal{O} \left(\frac{2^{c\sqrt{\log{n}}}}{n^{8/7}}\right)$$
This has been recently claimed to be improved by Cohen, Pohoata and Zakharov \cite{cohen2023new} to $\mathcal{O}\left(n^{-8/7-1/2000}\right)$.

There has also been work on several variants of this problem. Jiang, Li and Vitany \cite{DBLP:journals/rsa/JiangLV02} and Benevides, Hoppen, Lefmann and Odermann \cite{benevides2023heilbronn} studied the case in which the points were randomly distributed. The problem was also explored in higher dimensions by placing $n$ points in $d$-dimensional unit cubes $[0,1]^d$ instead of a unit square \cite{DBLP:journals/siamdm/Barequet01,DBLP:journals/dm/Barequet04,DBLP:journals/dcg/BarequetS07, DBLP:journals/siamdm/Brass05, DBLP:journals/combinatorica/Lefmann03, DBLP:journals/siamcomp/LefmannS02}. 

Schmidt asked about the value of $\Delta_k(n)$ and proved that $\Delta_4(n)=\Omega\left(\frac{1}{n^{1.5}}\right)$ \cite{10.1112/jlms/s2-4.3.545}. Bertraln-Kretzberg, Hofmeister and Lefmann generalized this result to $k$-gons by proving that $\Delta_k(n)=\Omega\left(\frac{1}{n^{\frac{k-1}{k-2}}}\right)$ \cite{DBLP:journals/siamcomp/Bertram-KretzbergHL00}. This was improved by Lefmann\cite{DBLP:journals/ejc/Lefmann08} to $$\Delta_k(n)=\Omega\left(\frac{{(\log{n})}^{1/k-2}}{n^{1+\frac{1}{k-2}}}\right)$$

\section{Our results}
A trivial upper bound of $\Delta_4(n) \leq \dfrac{3}{n}$ can be obtained by subdividing the unit square into squares of side length $\sqrt{\dfrac{3}{n}}$ using the pigeonhole principle. However, despite several efforts to improve this upper bound (asymptotically) since it was posed in $1972$, there has been no progress! Szemer{\'e}di asked if at least the constants in this upper bound can be improved \cite{szemeredi2022}. In this work, we answer this question by proving the following theorem. 

\begin{theorem}
\label{thm1}
$\Delta_4(n) \leq \dfrac{2}{n} + o\left(\dfrac{1}{n}\right)$    
\end{theorem}

We also generalize our result to general $k$-gons for any constant $k\geq 4$.  
\begin{theorem}
\label{thm2}
$\Delta_k(n) \leq \dfrac{k-2}{n} + o\left(\dfrac{1}{n}\right)$    
\end{theorem}

%We conjecture (the exact values of this function) that $\Delta_{k}(\alpha(k-1))=\dfrac{1}{2(\alpha-1)}$ for $2\leq \alpha \leq k-1$. We proved this conjecture for $\alpha=2$ %and are working towards proving it for $\alpha>2$.

%\newpage
\section{Convex quadrilaterals: Proof of Theorem \ref{thm1}}
We solve a more general problem by having the points on a unit rectangle (instead of a unit square). Given a set $\mathcal{P}=\{P_1,P_2,\ldots,P_n\}$ of $n\geq 3$ points on the unit rectangle $[0,d]\times [0,d^{-1}]$ and $k \leq n$, let $\mathcal{A'}_{k}(P)$ be the minimum area of the convex hull determined by a set of $k$ points in $\mathcal{P}$ for any $0<d\leq 1$. The supremum value of $\mathcal{A'}_{k}(P)$ over all choices of $\mathcal{P}$ is denoted by $\Delta_k'(n)$. It is easy to see that by definition $$\Delta_k(n) \leq \Delta_k'(n)$$
For $k=4$, when there are $n$ points, in the argument to obtain a trivial bound of $\Delta_4'(n) \leq \dfrac{3}{n-1}$, we partition the unit rectangle into at most ${(n-1)}/{3}$ smaller rectangles. This would guarantee that there exists a small rectangle containing at least $4$ points. Naturally, one can also extend this idea to make sure there are at most $\left( \dfrac{n-1}{n'-1} \right)$ smaller rectangles and force $n'\geq 4$ points into one rectangle. This gives us a relation between $\Delta_4'(n)$ and $\Delta_4'(n')$, which is formalized in Observation \ref{phpobs}.
%For any real number $x$, we use ${\floor{x}}$ to denote the largest integer, which is not greater than $x$ and ${\ceil*{x}}$ to denote the smallest integer, which is not less than $x$.
\begin{observation}
\label{phpobs} For $4\leq n'\leq n$,
$$\Delta_4'(n) \leq {\floor{\dfrac{n-1}{n'-1}}}^{-1}\Delta_4'(n')$$
\end{observation}
\begin{proof}
Partition the unit area into a grid with ${\floor{\dfrac{n-1}{n'-1}}}$ %\footnote{Ignoring the greatest integer function for large $n$} 
rectangles of area ${\floor{\dfrac{n-1}{n'-1}}}^{-1}$. Since there are $n$ points and ${\floor{\dfrac{n-1}{n'-1}}} \leq \dfrac{n-1}{n'-1}$ rectangles, by the pigeonhole principle, one of the smaller rectangles (with their boundary included) has at least $n'$ points. Therefore, there always exists $n'$ points within an area at most ${\floor{\dfrac{n-1}{n'-1}}}^{-1}$. It now follows by a scaling argument that there exist four points within an area at most ${\floor{\dfrac{n-1}{n'-1}}}^{-1}\Delta_4'(n')$.
\end{proof}
When $n'=4$, this gives the trivial bound of $\Delta_4'(n) \leq \dfrac{3}{n-3} \approx 3/n$ as expected. We can do slightly better by tuning the value of $n'$ to be $6$. We start with finding the exact value of $\Delta_4'(6)$.

%make the following observation to improve this to $\Delta_4 \leq \dfrac{2.5}{n-5} \approx 2.5/n$.

\begin{observation}
\label{6_vertex}
$\Delta_4'(6)=1/2$    
\end{observation}

\begin{proof}
Let $C$ be the centre of the rectangle $[0,d]\times [0,d^{-1}]$ and $\mathcal{P}$ be a set of any six points in $[0,d]\times [0,d^{-1}]$. Pick an arbitrary point $P \in \mathcal{P}$ and extend the line segment $PC$ into a line. The extended line $PC$ cuts the rectangle $[0,d]\times [0,d^{-1}]$ into two convex parts, and by symmetry, both these parts have the same area, i.e., $1/2$. From the pigeonhole principle, at least $3$ of the remaining $5$ points lie on one side of the extended line $PC$. Therefore, $4$ points (including $P$) exist, which are contained in a convex shape whose area is $1/2$. Therefore, $\Delta_4'(6)\leq 1/2$.

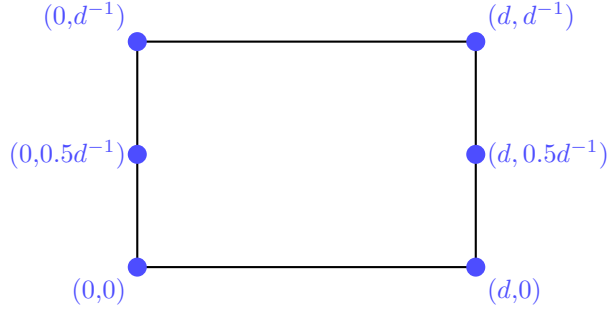
\begin{figure}[htbp]
    \centering

\begin{tikzpicture}[scale=3, line width=0.7pt]
  % Draw the unit square
  \draw[thick] (0,0) rectangle (1.5,1);
  
  % Highlighted points
  \fill[blue!70] (0,0) circle (1.2pt) node[below left, font=\footnotesize] {(0,0)};
  \fill[blue!70] (0,0.5) circle (1.2pt) node[left, font=\footnotesize] {(0,$0.5d^{-1}$)};
  \fill[blue!70] (0,1) circle (1.2pt) node[above left, font=\footnotesize] {(0,$d^{-1}$)};
  \fill[blue!70] (1.5,0) circle (1.2pt) node[below right, font=\footnotesize] {($d$,0)};
  \fill[blue!70] (1.5,0.5) circle (1.2pt) node[right, font=\footnotesize] {($d,0.5d^{-1}$)};
  \fill[blue!70] (1.5,1) circle (1.2pt) node[above right, font=\footnotesize] {($d,d^{-1}$)};
\end{tikzpicture}
  \caption{$\Delta_4'(6)\geq 1/2$}
    \label{fig:p6}
\end{figure}

It may be noted the bound of $1/2$ is achieved in Figure \ref{fig:p6}. Therefore $\Delta_4'(6)\geq 1/2$
\end{proof}

\begin{corollary}
\label{6_result}
$\Delta_4(n) \leq \dfrac{1}{2} \floor{\dfrac{n-1}{5}}^{-1} \leq \dfrac{2.5}{n-5}$    
\end{corollary}
\begin{proof}
By substituting $n'=6$ in Observation \ref{phpobs} and using Observation \ref{6_vertex}, we get 
$$\Delta_4'(n) \leq {\floor{\dfrac{n-1}{5}}}^{-1}\Delta_4'(6) = \dfrac{1}{2}{\floor{\dfrac{n-1}{5}}}^{-1} \leq \dfrac{2.5}{n-5}$$

\end{proof}
We will now extend the idea of Observation \ref{6_vertex} to all $n$ of the form $2^s+2$ for any $s\geq 2$, i.e.,  for all $n$ of such form, we prove $\Delta_4(n)\leq \dfrac{2}{n-2}$ in Theorem \ref{recurse_theorem} using Observation \ref{cut_two} and Lemma \ref{cut_lemma}.

%We start with making some simple observations.

\begin{observation}
\label{cut_two}
For any point $P$ in a convex polygon $C$, there exists a line through $P$ which partitions $C$ into two halves of equal area.    
\end{observation} 
\begin{proof}
Let $A$ be the area of $C$. Pick any arbitrary line $L$ through $P$ and let the areas of the convex polygons on both of its sides of $L$ be $L_1$ and $L_2$ such that $L_1\leq A/2 \leq L_2$. By rotating the line by ${180}^{\circ}$ around $P$, we get $L_2 \leq 1/2 \leq L_1$. Since $L_1$ changes continuously as a function of the angle of rotation $\theta$, by the intermediate value theorem, it must have achieved $1/2$ in between for some $\theta$.
\end{proof}

\begin{lemma}
\label{cut_lemma}
If a convex polygon of area $\Delta$ has $2^i\cdot \beta+2$ points, then there is a convex polygon of area $\Delta/2$ which contains at least $2^{i-1}\cdot \beta+2$ points.
\end{lemma}
\begin{proof}

Pick one of the $2^i\cdot \beta+2$ points arbitrarily, say $P$. From Observation \ref{cut_two}, there is a line $L$ through $P$ cutting the polygon into two halves. Note that by the pigeonhole principle, one of the halves would have at least $2^{i-1}\cdot \beta+1$ points. By including $P$ in this region, we get a convex polygon of area $\Delta/2$, containing at least $2^{i-1}\cdot \beta+2$ points.
\end{proof}

 %Therefore, it is sufficient to prove $\Delta_(n) \leq 2/n$.

We first introduce some new notation. Given a set $\mathcal{Q}=\{Q_1,Q_2,\ldots,Q_n\}$ of $n\geq 3$  points on an \textit{arbitrary convex object $\mathcal{C}$ of unit area}, let $\mathcal{A}_k^{\mathcal{C}}(Q)$ be the minimum area of the convex hull determined by some $k$ points in $\mathcal{Q}$. The supremum value of $\mathcal{A}_k^{\mathcal{C}}(Q)$ over all choices of $\mathcal{Q}$ is denoted by $\Delta^{\mathcal{C}}_k(n)$. Let $\Delta_k''(n)$ denote the supremum value of $\Delta^{\mathcal{C}}_k(n)$ over all convex objects of unit area. It is easy to see that by definition $$\Delta_k(n) \leq \Delta_k'(n) \leq \Delta_k''(n)$$
\begin{theorem}
\label{recurse_theorem}
If $n=2^s+2$ for some integer $s>0$, then $$\Delta_4(n) \leq \Delta_4'(n) \leq \dfrac{1}{2^{s-1}} = \dfrac{2}{n-2}$$   
\end{theorem}

\begin{proof}
Let $\mathcal{Q}=\{Q_1,Q_2,\ldots,Q_n\}$ be a set of $n$ points. We will prove a stronger statement of $$\Delta''_4(n) \leq \dfrac{1}{2^{s-1}} $$

\begin{observation}
\label{obsrec}
For every $i \in [0,s-1]$, there exists a convex polygon of area at most $\dfrac{1}{2^{i}}$ which has at least $2^{s-i}+2$ points from $\mathcal{Q}$.
\end{observation}
\begin{proof}
We prove this by induction on $i$. When $i=0$, the claim is true by definition. Suppose there exists a convex polygon of area at most $\dfrac{1}{2^{i}}$ with at least $2^{s-i}+2$ points from $\mathcal{Q}$, then from Lemma \ref{cut_lemma}, for $i<s-1$, there exists a convex polygon of area at most $\dfrac{1}{2^{i+1}}$ with at least $2^{s-(i+1)}+2$ points from $\mathcal{Q}$ 
\end{proof}
By substituting $i=s-1$ in Observation \ref{obsrec}, Theorem \ref{recurse_theorem} follows.
\end{proof}

One may note that this would give an upper bound of $\approx 2/n$ for many arbitrarily large $n$ of the form $2^s+2$. However, there are also several arbitrarily large $n$ of the form, say, $2^s+1$, for which this bound is not useful. We fix this using Observation \ref{phpobs}.

\begin{corollary}
$\Delta_4(n) \leq \dfrac{2}{n}+o\left(\dfrac{1}{n}\right)$    
\end{corollary}
\begin{proof}
For every $n$, there exists some $i$ such that $$2^{i-1}+2\leq n<2^{i}+2$$
Let $n'=2^{\ceil*{0.5i}}+2$. From Theorem \ref{recurse_theorem}, $$\Delta_4'(n')\leq \dfrac{2}{n'-2} < \dfrac{2}{n'-1} $$

$${\floor{\dfrac{n-1}{n'-1}}} > \frac{n-1}{n'-1}-1 \geq \frac{n-n'}{n'-1} $$
From Observation \ref{phpobs}, $$\Delta_4'(n) \leq \Delta_4'(n'){\floor{\dfrac{n-1}{n'-1}}}^{-1} \leq \dfrac{2}{n'-1}\cdot \frac{n'-1}{n-n'} = \dfrac{2}{n-n'} $$
$$=\dfrac{2}{n} + \dfrac{2n'}{n(n-n')} =\dfrac{2}{n}+\mathcal{O}\left(\dfrac{1}{n^{1.5}}\right)$$
%We first prove that there is some large $\alpha$ of the form $2^s+2$ such that $(\alpha-1)\Delta_4(\alpha)\approx 2$. From Theorem \ref{recurse_theorem}, we know that $(\alpha-1) \Delta_4(\alpha)=\dfrac{2^s+1}{2^{s-1}}=2+2^{-{s-1}}$. Therefore, for sufficiently large $\alpha$, $(\alpha-1)\Delta_4(\alpha)\approx 2$. 

%By substituting $t=\alpha$ in Observation \ref{phpobs}, we get $\Delta_4(n) \leq 2/n$ for all large $n$ (sufficiently greater than $\alpha$).
\end{proof}

\begin{figure}[htbp]
    \centering

\begin{tikzpicture}[scale=3, line width=0.7pt]
  % Draw the unit square
  \draw[thick] (0,0) rectangle (1.5,1);
  
  % Highlighted points
  \fill[blue!70] (0,0) circle (1.2pt) node[below left, font=\footnotesize] {(0,0)};
  \fill[blue!70] (0,0.5) circle (1.2pt) node[left, font=\footnotesize] {(0,$0.5d^{-1}$)};
  \fill[blue!70] (0,1) circle (1.2pt) node[above left, font=\footnotesize] {(0,$d^{-1}$)};

  \fill[blue!70] (0.75,0) circle (1.2pt) node[below left, font=\footnotesize] {($0.5d$,0)};
  \fill[blue!70] (0.75,0.5) circle (1.2pt) node[left, font=\footnotesize] {($0.5d$,$0.5d^{-1}$)};
  \fill[blue!70] (0.75,1) circle (1.2pt) node[above left, font=\footnotesize] {($0.5d$,$d^{-1}$)};

  \fill[blue!70] (1.5,0) circle (1.2pt) node[below right, font=\footnotesize] {($d$,0)};
  \fill[blue!70] (1.5,0.5) circle (1.2pt) node[right, font=\footnotesize] {($d,0.5d^{-1}$)};
  \fill[blue!70] (1.5,1) circle (1.2pt) node[above right, font=\footnotesize] {($d,d^{-1}$)};
\end{tikzpicture}

  \caption{$\Delta_4'(9)\geq 1/4$}
    \label{fig:p9}
\end{figure}
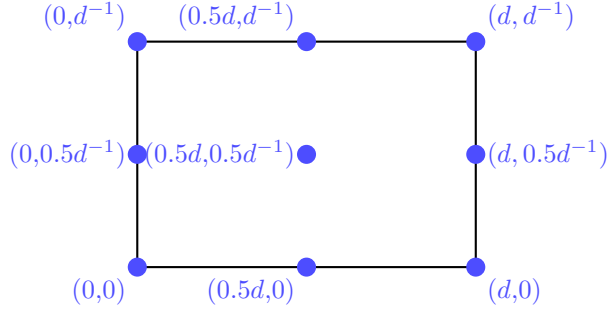

From Figure \ref{fig:p9}, it is easy to see that $\Delta_4'(9)\geq 1/4$. We conjecture that the other direction is also true, i.e., $\Delta_4'(9)\leq 1/4$

\begin{conjecture}
\label{conj1}
$\Delta_4'(9)=1/4$
\end{conjecture}

If true, Conjecture \ref{conj1} would directly imply that $\Delta_4(n) \leq \dfrac{2}{n-8}$ from Observation \ref{phpobs} by picking $n'$ to be $9$. We also note that finding the exact values of $\Delta_k{(n)}$ is of independent interest and has been well studied for $k=3$ \cite{comellas2002new, dehbi2022heilbronn, zeng2008heilbronn}. 

One may note that our analysis will extend to general convex figures with unit area (instead of just unit squares), i.e., 
$$\Delta''_4(n) \leq \dfrac{2}{n} + o\left(\dfrac{1}{n}\right)$$

\section{Convex $k$-gons: Proof of Theorem \ref{thm2}}
In this section, we extend our proofs for $k=4$ to give upper bounds on $\Delta_k(n)$ for any constant $k$. %A trivial upper bound is $\dfrac{k-1}{n}$ using grids and pigeon hole principle. We improve it to $\dfrac{k-2}{n}$. 

\begin{proposition}[Analogue of Observation \ref{phpobs}]
\label{phpobs2} For $n'\leq n$,
$$\Delta_k'(n) \leq {\floor{\dfrac{n-1}{n'-1}}}^{-1}\Delta_k'(n')$$
\end{proposition}
\begin{proof}
The proof of Proposition \ref{phpobs2} directly extends from Observation \ref{phpobs}.
\end{proof}

\begin{theorem}[Analogue of Theorem \ref{recurse_theorem}]
\label{recurse_theorem2}
If $n=2^s(k)-2^{s+1}+2$ for some integer $s>0$, then $$\Delta_k(n) \leq \Delta_k'(n) \leq \dfrac{1}{2^{s}} = \dfrac{k-2}{n-2}$$   
\end{theorem}

\begin{proof}
Let $\mathcal{Q}=\{Q_1,Q_2,\ldots,Q_n\}$ be a set of $n$ points. We will prove a stronger statement of $$\Delta''_k(n) \leq \dfrac{1}{2^{s}} $$

\begin{observation}
\label{obsrec2}
For every $i \in [0,s]$, there exists a convex polygon of area at most $\dfrac{1}{2^{i}}$ with at least $2^{s-i}(k)-2^{s+1-i}+2$ points from $\mathcal{Q}$.
\end{observation}
\begin{proof}
We prove this by induction on $i$. When $i=0$, the claim is true by definition. Suppose there exists a convex polygon of area at most $\dfrac{1}{2^{i}}$ with at least $2^{s-i}(k)-2^{s+1-i}+2$ points from $\mathcal{Q}$, then from Lemma \ref{cut_lemma}, for $i<s$, there exists a convex polygon of area at most $\dfrac{1}{2^{i+1}}$ with at least $2^{s-(i+1)}(k)-2^{s+1-(i+1)}+2$ points from $\mathcal{Q}$.
\end{proof}
By substituting $i=s$ in Observation \ref{obsrec2}, Theorem \ref{recurse_theorem2} follows.
\end{proof}

\begin{corollary}
$\Delta_k(n) \leq \dfrac{k-2}{n}+o\left(\dfrac{1}{n}\right)$    
\end{corollary}
\begin{proof}
For every $n$, there exists some $i$ such that $$2^{i-1}(k)-2^{i}+2 \leq n <2^i(k)-2^{i+1}+2$$
Let $n'=2^{\ceil*{0.5i}}(k)-2^{\ceil*{0.5i}+1}+2$. From Theorem \ref{recurse_theorem2}, $$\Delta_k'(n')\leq \dfrac{k-2}{n'-2} < \dfrac{k-2}{n'-1} $$

$${\floor{\dfrac{n-1}{n'-1}}} > \frac{n-1}{n'-1}-1 \geq \frac{n-n'}{n'-1} $$
From Proposition \ref{phpobs2}, $$\Delta_k'(n) \leq \Delta_k'(n'){\floor{\dfrac{n-1}{n'-1}}}^{-1} \leq \dfrac{k-2}{n'-1}\cdot \frac{n'-1}{n-n'} = \dfrac{k-2}{n-n'} $$
$$=(\dfrac{1}{n} + \dfrac{n'}{n(n-n')})\cdot (k-2) =\dfrac{k-2}{n}+\mathcal{O}\left(\dfrac{1}{n^{1.5}}\right)$$
\end{proof}

\begin{conjecture}
\label{gen_conjecture}
$\Delta_{k}'(\alpha(k-1))=\dfrac{1}{2(\alpha-1)}$ for $\alpha \leq k-1$
\end{conjecture}
One can prove that $\Delta_{k}'(\alpha(k-1)) \geq \dfrac{1}{2(\alpha-1)}$ by placing the points in the corners of a $\alpha \times (k-1)$ grid. We conjecture that this is, in fact, the optimal placement. When $\alpha=2$, this is indeed true by Lemma \ref{cut_lemma}. Conjecture \ref{conj1} is a special case of Conjecture \ref{gen_conjecture} when $k=4$. If true, Conjecture \ref{gen_conjecture} would imply an upper bound of $\approx {k}/{(2n)}$

One may notice that Conjecture \ref{gen_conjecture} can not be extended to $\alpha \leq k-1$, as at least $k$ points become collinear in that case.

\section*{Acknowledgements}
RG thanks Endre Szemer{\'e}di for introducing this problem to him during the $9$th Heidelberg Laureate Forum. RG thanks Saladi Rahul for his comments on the manuscript.

%%%%%%%%%%%%%%%%%%%%%%%%%%%%%%%%%%%%%%%%%%%%%%%%%%%%%%
%%%%%%%%%%%%%%%%%%%%%%%%%%%%%%%%%%%%%%%%%%%%%%%%%%%%%%

%------------------------------------------------------------------
%------------------------------------------------------------------

%------------------------------------------------------------------
%------------------------------------------------------------------

%-------------------------------------------------------------------------

\BibTexMode{%
   \bibliographystyle{abbrv}
   \bibliography{template.bib}
}%
\BibLatexMode{\printbibliography}

\end{document}